\documentclass{article}

\usepackage[english]{babel}

\usepackage[letterpaper,top=2cm,bottom=2cm,left=3cm,right=3cm,marginparwidth=1.75cm]{geometry}

\usepackage{amsmath}
\usepackage{amssymb}
\usepackage{amsthm}
\usepackage{enumerate}   
\usepackage{graphicx}
\usepackage[colorlinks=true, allcolors=blue]{hyperref}
\usepackage{authblk}
\usepackage{fullpage}
\usepackage{xcolor}

\usepackage[normalem]{ulem}

\newtheorem{theorem}{Theorem}[section]
\newtheorem*{theorem*}{Theorem}
\newtheorem{corollary}[theorem]{Corollary}
\newtheorem{definition}[theorem]{Definition}
\newtheorem{lemma}[theorem]{Lemma}
\newtheorem*{lemma*}{Lemma}
\newtheorem{proposition}[theorem]{Proposition}

\newtheorem{remark}[theorem]{Remark}

\begin{document}

\title{For how long time evolution of chaotic or random systems can be predicted.}

\author[1]{Leonid Bunimovich}
\author[1, 2]{Kirill Kovalenko}
\affil[1]{School of Mathematics, Georgia Institute of Technology, Atlanta, USA}
\affil[2]{Scuola Superiore Meridionale, School for Advanced Studies, Naples, Italy}
\date{}

\maketitle

\section{Introduction}
Traditionally, the theories of random processes and of chaotic dynamical systems studied only asymptotic in time properties of observables in such systems. Likewise, only averages of observables over infinite time intervals were considered in these theories. Such problems seem to be the only ones tractable, and, moreover, the only ones that make sense to consider. Indeed, the questions like what are the differences between the evolution of the system in two different instances of time seem to be not only non-tractable, but in fact unreasonable. It turned out that there are meaningful questions about finite time evolution of random and chaotic dynamical systems, which not just make sense and are interesting, but, moreover, such questions can be exactly answered. 
A very basic question in this new area is for how long is it possible to make finite-time predictions. It turned out that the longest predictions can be made for the most random (chaotic) systems like a coin tossing, although it seems to be rather counterintuitive. Indeed, it is naturally to expect that for the most random systems just finite time predictions can be made for the shortest time intervals. However, the situation is totally opposite. In fact longer finite time predictions can be made for systems with the simplest dynamics, rather than for weakly chaotic systems, where dynamics itself is much more nonuniform and complicated.
A principal question in this area is how the possible length of finite time predictions depends upon the precision of observing dynamics (time evolution). In other words, what is going to happen to a possible length of finite time prediction if we consider dynamics under a more detailed description of the states of a system. One may think about that as considering more and more refined partitions of the phase space of the system in question.
Again, it seems to be natural that predictions in case when we consider more detailed characteristics of the system could be made for shorter times. However, it turned out that the situation is quite opposite \cite{BoldingBunimovich2019}, i.e., in fact, the length of the time interval, where finite time predictions are possible, is growing at least linearly with the size of an alphabet used to construct a partition of the phase space. In the present paper, we show, by a totally different analytic approach, that the length of a time interval, where predictions about a character of dynamics are possible. In fact, this time for the "simplest" random and chaotic systems grows exponentially with respect to the length of an alphabet, when we consider transport of orbits in phase spaces of such systems in more and more details.

\section{Some basic definitions and notations}

Our setup is quite similar to the one in the paper\cite{BoldingBunimovich2019}. Namely, we consider an ergodic dynamical systems that preserve absolutely continuous measure $\nu$. We also assume that these systems have  finite Markov partitions. Moreover, we assume that the corresponding symbolic representation is a full Bernoulli shift, i.e., all elements of a Markov partition have equal measures and all transition probabilities are equal each other. (In what follows, we will refer to this partition as a \textit{basic Markov partition}). Hence, if the partition consists of $q$ elements, then all transition probabilities are equal to $1/q$. Such dynamical systems were called \textit{fair-dice-like} (FDL) systems since the dynamic is equivalent to throwing a fair dice with $q$ faces. 

One of the simplest examples of an FDL-system is a doubling map of the unit interval $x \to 2x \$(mod{1})$,  with the basic Markov partition of $[0,1]$ into $(0,1/2)$ and $(1/2,1)$. 
Given a Markov partition $\mathcal{P}$, one can consider its \textit{refinement of the $m$-th order} by taking the intersection of all the preimages of all elements of the Markov partition from $0$-th to $(m-1)$-st order. Thus, such partitions are generated by $C^{i_0}_{\mathcal{P}} \cap T^{-1}C^ {i_1}_{\mathcal{P}}\cap T^{-2}C^{i_2}_{\mathcal{P}}\cap\dots\cap T^{m-1}C^{i_{m-1}}_{\mathcal{P}}$, where $C^{i_{k}}$ are the elements of the basic Markov  partition. It is well known that any refinement of a Markov partition is also a Markov partition.
Following \cite{Bunimovich2012} and \cite{BoldingBunimovich2019}, we present now the following 
\begin{definition}
    An uniformly hyperbolic dynamical system preserving Borel probability measure $\mu$ is called a fair-dice-like (or FDL) system if there exists a finite Markov partition $\mathcal{P}$ of its phase space $M$ such that for any integers $m$ and $\{i_k : 1 \leq i_k\leq  q\}_{k=0}^{m}$ one has
    \begin{equation*}
        \mu\left(C^{i_0}_{\mathcal{P}} \cap T^{-1}C^{i_1}_{\mathcal{P}}\cap T^{-2}C^{i_2}_{\mathcal{P}}\cap\dots\cap T^{m-1}C^{i_{m-1}}_{\mathcal{P}}\right)=\frac{1}{q^m}
    \end{equation*}
    where $q$ is the number of elements in the
partition $\mathcal{P}$ and $C^{i}_{\mathcal{P}}$ is the element number $i$ of $\mathcal{P}$.
\end{definition}

Hence, an FDL-system is a full Bernoulli shift with equal transition probabilities between states. Denote the elements of the basic Markov partition $\mathcal{P}$ by letters of an alphabet $\Omega$ of size $q$. The elements of the $k$-th refinement of the basic partition could then be uniquely coded by sequences of $k$ letters of $\Omega$. In what follows, we will call these codes words or strings. 

It provides a symbolic representation of the FDL-systems. Particularly, almost all points of the phase space $M$ have a corresponding unique representation by infinite sequences of letters from $\Omega$ corresponding to iterations $T^{i}$ of our map.
Consider a uniform probability distribution on the set of infinite sequences of the letters in our alphabet. Take now any element of a refined Markov partition, and call it a hole. This means that all points that visit this element in a course of dynamics escape from the phase space. Therefore, the question about the probability of escape at time $n$ through a hole, which is an element of a refined Markov partition, could be reformulated in terms of probabilities of the set of strings, which encounter for the first time a word corresponding to that element. More formally, let $h_w(n)$ be the number of strings of length $n$, which contain the word $w$ of length $k$ as their last $k$ characters, and do not have $w$ as a substring of $k$ consecutive characters in any other part of this string. Then the probability of escape exactly at the time $n$ through a hole corresponding to the word $w$ equals $P_h(w,n) = h_w(n+k)/q^{n+k}$. This probability is called the first hitting probability corresponding to the word  $w$. 
Our main goal is to analyze the relations between the first hitting probabilities for different words $w$ having the same length $k$. In more general terms, we are interested in 
characterizing finite-time transport of orbits in the phase spaces of dynamical systems. By comparing $h_w(n)$ for different words $w$, we will be able to identify which elements of a refinement of our Markov partition, and thus the points in the corresponding elements, spread over the phase space faster than the other elements.

\section{Main Results}

In order to understand the behavior of $h_w(n)$, it is convenient to use the notion of \textit{autocorrelation} of words introduced by J.Conway.
\begin{definition}
    Given a word $w$ of length $k$ the autocorrelation of $w$ is a sequence of $k$ numbers $b_k, b_{k-1},\dots,b_1$ where $b_i$ is equal to $1$ if the prefix of $w$ of the length $i$ coincides with the suffix of $w$ of the same length $i$, and $0$ otherwise.
\end{definition}
For example, the autocorrelation of the word $HTHTHT$ is $101010$, and the autocorrelation of $HTTHTH$ is $100001$. The autocorrelation of a word is related to its "internal periodicity". For instance,  the word $HTHTHT$ is generated by the repetition of the pattern $(HT)$. Therefore,the autocorrelation of $HTHTHT$ is $101010$, since for every even $i$ both the prefix of the size $i$ and the suffix of the size $i$ consist of $HT$ repeated $i/2$ times, while for odd $i$ the prefix starts with $HT$, and there is a mismatch because the suffix starts with $TH$. 

 We can compare (values of) autocorrelations by considering them as numbers written in the base $2$. For instance, using the example above, the autocorrelation of $HTHTHT$ is greater than the autocorrelation of $HTTHTH$ since the sequence $101010$ becomes $42$ and the sequence $100001$ becomes $33$.

 It turns out that autocorrelation can be effectively used to determine how many strings of a specified length either avoid a given word entirely or encounter it for the first time at the end of the string. In more detail, defined now autocorrelation polynomials provide an effective tool to answer this question. Given a word $w$ with the autocorellation $b_k b_{k-1}\dots b_1$ the \textit{autocorrelation polynomial} of the word $w$ is defined as

\begin{equation*}
    A_{w,z}(z)=b_k z^0+\dots+b_{1}z^{n-1}.
\end{equation*}

Then, $h_w(n)$, which is the number of words of length $n$ containing the word $w$ only once and at the end, has the following generating function (see Proposition I.4. in \cite{Flajolet_Sedgewick_2009}):

\begin{equation*}
    GF_{w,z}(z) = \sum_{n=0}^{\infty} h_w(n)z^n = \frac{z^k}{z^{k}+(1-q z)A_{w,z}(z)}.
\end{equation*}

For a further analysis, it is convenient to make the following change of variable $x=\frac{1}{z}$. 

Then one gets

\begin{equation*}
    GF_w(x)=GF_{w,z}(1/x)=\frac{1}{1+(x-q)A_{w}(x)} 
\end{equation*},

where in what follows

$$
A_{w}(x) = b_k x^{k-1} + b_{k-1}x^{k-2} +\dots+b_{1}x^{0} = x^{k-1} + b_{k-1}x^{k-2}+\dots+b_{1}x^{0},
$$

which is an \textit{inverse autocorrelation polynomial}. 
Observe that, in particular, for any two words $w$ and $w'$ with equal autocorrelations the corresponding generating functions $GF_w(x)$ and $GF_{w'}(x)$ are also equal. Hence $h_w(n) = h_{w'}(n)$ for all $n$. For different autocorrelations, however, the autocorrelation polynomials and the generating functions are different.  Thus there is at least one $n$, such  that $h_w(n)\neq h_{w'}(n)$. 

Given $w$, one can consider a piecewise-linear function, which is equal to $P_h(w,n)$ for any positive integer $n$, and it is linear between any  $n$ and $n+1$. In what follows, we will call this function a \textit{first hitting probability curve for $w$}. Since the dynamical systems in question are ergodic, the sum of all the first hitting probabilities is equal to $1$, i.e. $\sum_{n=0}^{\infty}P_h(w,n)=1$. Hence, for two different words $w$ and $w'$ the corresponding first hitting probabilities curves either coincide or intersect at least once. Equivalently, we have for some $n$ 
\begin{equation*}
    P_h(w,n) \leq P_h(w',n) \text{ and such that } P_h(w,n+1) > P_h(w',n+1)
\end{equation*}
or vice versa. 

One may naturally think that for an ergodic system any two such curves should intersect infinitely many times.
It turns out, however, that for any $w$ and $w'$ with different autocorrelations $P_h(w,n)$ and $P_h(w',n)$ intersection occurs \textit{exactly once}:

\begin{theorem*}\cite{BoldingBunimovich2019}]
    Consider an FDL-system. Let $w$ and $w'$ be words coding some elements of (possibly different) refinements of the basic Markov partition such that the autocorrelation of $w$ is greater than the autocorrelation of $w'$. Then there exists an $N > k$ such that $P_h(w,n) - P_h(w',n)\leq 0$ for $n < N$, and $P_h(w,n) - P_h(w',n)> 0$ for $n > N$.
\end{theorem*}

Therefore, the entire time interval gets partitioned into three intervals: the first interval where there are no intersections between different first-hitting probability curves. Therefore in this interval there is a hierarchy of these curves, where $P_w(n) \leq P_{w'}(n)$ if $A_w(2)> A_{w'}(2)$; the second interval where the all the intersections between the first probability curves occur, and therefore there is no hierarchy of the first hitting probability curves in this interval; and the third interval, starting after all the intersections occurred, where there is the inverse hierarchy of these curves. An example for $k=4$ is given in the Fig.\ref{fig:Intervals}.

\begin{figure}
    \centering
    \includegraphics[width=0.85\linewidth]{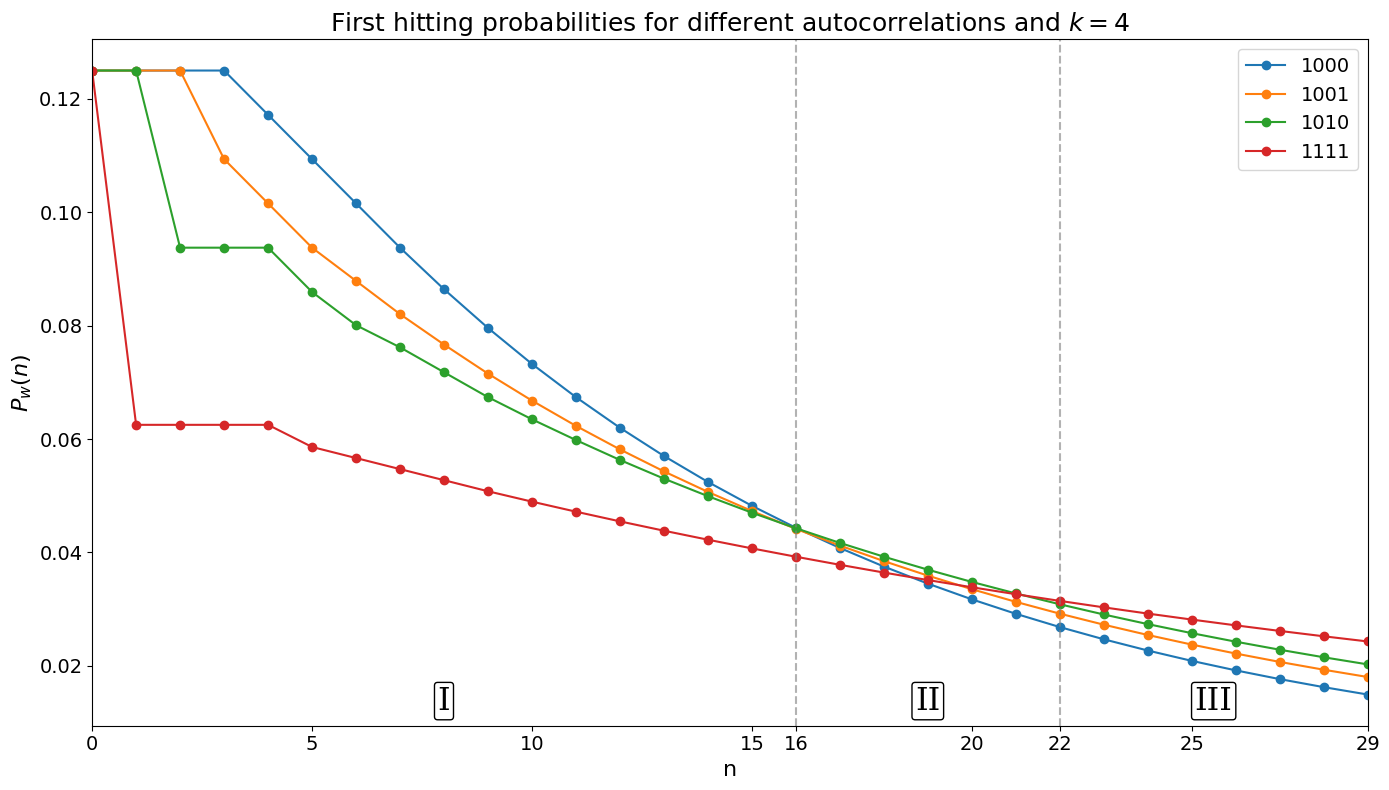}
    \caption{First hitting probabilities curves for all different possible autocorrelations of the words of the size $k=4$ (color coded as shown in the legend). The horizontal axis is time $n$, and the vertical axis is the first hitting probability $P_w(n)$. One can see that the time axis can be separated into three intervals -  the first one from $0$ to $16$ with the fixed hierarchy of probabilities; the second one from $16$ to $22$ when all the intersections happen; and the third one from $22$ to infinity with the hierarchy opposite to the one in the first time interval.}
    \label{fig:Intervals}
\end{figure}

Observe that the hierarchy in the third interval matches the hierarchy obtained by comparing the commonly used escape rates \cite{BruinDemersTodd2018} of the elements of the refined Markov partition. Therefore, it is of fundamental importance to understand the length of the first time interval, since only in this interval finite-time predictions are possible.

In \cite{BoldingBunimovich2019} it was demonstrated that the  length of the first time interval grows at least linearly in $k$, as one considers more and more refined partitions of the phase space. We will show that the length of the first  interval actually grows exponentially. Thus, our main result is the following one

\begin{theorem} \label{th:main}
    Given $0<\delta\leq 0.25$ there is $k_{\delta}$ such that for any two words $w$ and $w'$ of the same length $k\geq k_{\delta}$, where autocorrelation of $w$ is larger than the autocorrelation of $w'$, and $A_w(x)$ and $A_{w'}(x)$ are the inverse autocorrelation polynomials for $w$ and $w'$ correspondingly, the time $N$ of intersection between $h_w(n)$ and $h_{w'}(n)$ satisfies the following inequalities 
    \begin{equation} \label{eq: main-th bounds on N}
       (1-2\delta)\cdot C_{A_w/A_{w'}}\cdot q A_w(q) < N < (1+2\delta)\cdot C_{A_w/A_{w'}}\cdot q A_w(q)
    \end{equation}
    where 
    \begin{equation*}
    C_{A_w/A_{w'}} = \ln\left(1 + \frac{A_{w}(q)-A_{w'}(q)}{A_{w'}(q)}\right)\Big/\left(\frac{A_{w}(q)-A_{w'}(q)}{A_{w'}(q)}\right)
    ~~\text{ and }~~ 
    \frac{\ln q}{q} \leq C_{A_w/A_{w'}}\leq 1.
    \end{equation*}
\end{theorem}

Hence, the lower bound on the rate of exponential growth of the length of the first time interval follows directly from this Theorem \ref{th:main}:
\begin{corollary}
    Consider $k_{0.25}$ from the Theorem \ref{th:main}. Then for any two words $w$ and $w'$ of the same length $k\geq k_{0.25}$ and with different autocorrelations the time $N$ of intersection between the curves $h_w(n)$ and $h_{w'}(n)$ is at least so large as
     \begin{equation*}
       N> \frac{\ln q}{2} \cdot q^{k-1}.
    \end{equation*}
    
\end{corollary}

\begin{remark}
    For a fixed $\delta$ (e.g. for $\delta=0.1$), by making a proof much more technical, one can show that the bound (\ref{eq: main-th bounds on N}) holds at least for all $k>21$ if $q=2$, at least for all $k>12$ if $q=3$, and  at least for $k>9$ if $q=4$ etc, which provide better bounds for larger $q$.  
    
    However, the extensive numerical computations demonstrated that the minimal $k$, for which the bound is true, in fact, is much smaller. For instance, if $q=2$, then already for $k=6$ and for all possible pairs of words $w$ and $w'$ with different autocorrelations, the maximal deviation of the intersection time $N$ from $C_{A_w/A_{w'}}\cdot q A_w(q)$ is $7.71\%$, which is actually decreasing after that. Starting with $k=13$, it is smaller than $1\%$. If $q=3$, then already for $k=7$ it equals $0.25\%$, and for $q=4$ this deviation becomes smaller than $1\%$ starting with $k=4$. 

    This difference in the minimal value of $k$ is due to a very conservative bound in Proposition \ref{proposition: bound on coefs other exponents} for the values of non-dominating exponents. By improving these estimates, one could obtain a proof for smaller values of $k$ as well.
\end{remark}

\section{Proof of the Main Theorem \ref{th:main}}

Recall that for a given $w$ the generating function of $h_w(n)$ is
\begin{equation*}
    GF_{w,z}(z)=\frac{z^k}{z^{k}+(1-q z)A_{w,z}(z)},
\end{equation*}

After a change of variables $x=\frac{1}{z}$ it becomes

\begin{equation*}
    GF_w(x)=GF_{w,z}(1/x)=\frac{1}{1+(x-q)A_{w}(x)}, 
\end{equation*}

where $A_{w}(x)$ is the inverse autocorrelation polynomial of $w$.
Denote  $P_w(x)={1+(x-q)A_{w}(x)}$, and call it \textit{a recurrence polynomial of $w$}. Let $\{\alpha_1,\dots,\alpha_r\}$ be all the distinct roots of $P_w(x)$ with the corresponding multiplicities $\{m_1,\dots, m_r\}$. Then, the Partial fraction decomposition of $GF_w(x)$ represents it as the sum of the following terms

\begin{equation} \label{eq:single_reciprical_term}
    \frac{c_{ij}}{(x - \alpha_i)^j} = \frac{c_{ij}}{(1/z - \alpha_i)^j} =  \frac{z^jc_{ij}}{(1 - \alpha_i z)^j}
\end{equation}

for $i=1\ldots r$ and $j=1\ldots m_i$. For $j=1$ the expression (\ref{eq:single_reciprical_term}) expands to a geometric progression:

\begin{equation*} 
    c_{ij}\sum_{n=0}^{\infty}(\alpha_i)^n z^{n+1} = c_{ij}\sum_{n=1}^{\infty}(\alpha_i)^{n-1}z^{n}.
\end{equation*}

For $j\geq 2$ (by following Section 5.4 in \cite{graham94}) one can rewrite (\ref{eq:single_reciprical_term}) as

\begin{multline*} 
    \frac{z^jc_{ij}}{(1 - \alpha_i z)^j} = c_{ij}\sum_{n=0}^{\infty}\binom{j-1+n}{j-1}(\alpha_i)^n z^{n+j} = c_{ij}\sum_{n=0}^{\infty}\frac{(n+1)\ldots(n+j-1)}{(j-1)!}(\alpha_i)^n z^{n+j} = \\
    = c_{ij}\sum_{n=j}^{\infty}\frac{(n-j+1)\ldots(n-1)}{(j-1)!}(\alpha_i)^{n-j} z^{n}.
\end{multline*}.

Consequently, the coefficient of $z^{n}$ in the series expansion of the generating function $GF_{w,z}(z)$ is

\begin{equation} \label{eq: combinatoric expression of h_w(n)}
    h_w(n)=\sum_{\alpha_i}\sum_{j=1}^{m_i} c_{ij}\binom{n-1}{j-1}(\alpha_i)^{n-j}.
\end{equation}

Therefore, to describe the behavior of the function $h_w(n)$, we need to estimate the roots of $P_w(x)$ together with their multiplicities, as well as the coefficients $c_{ij}$.
 
Our analysis is organized in the following way. First, given a word $w$, we obtain bounds on the roots of $P_w(x)$ and on the coefficients $c_{ij}$. Next, we derive sharper estimates for the largest root $\alpha_m$ of $P_w(x)$ and for the corresponding coefficient $c_{i_m 1}$.

Next, we show that for two words $w$ and $w'$ of the same length $k$ but with different autocorrelations, and for their corresponding largest roots $\alpha_m$ and $\alpha_{m}'$ and coefficients $c_{i_m 1}$ and $c_{i_{m}' 1}$, a moment of time $n$, when the graphs of functions $c_{i_m 1}(\alpha_m)^n$ and $c_{i_{m}' 1}(\alpha_{m}')^n$ intersect, scales exponentially with $k$. These exponents will be referred to as \emph{main or leading exponents} of the expansion (\ref{eq: combinatoric expression of h_w(n)}). Finally, we show that the contributions to (\ref{eq: combinatoric expression of h_w(n)}) coming from all other terms (\emph{non-dominant exponents}) are relatively negligible, and therefore the moment of intersection of the functions $h_w(n)$ and $h_{w'}(n)$ also scales exponentially with $k$.

\subsection{Non-dominant exponents of \texorpdfstring{$h_w(n)$}{hw(n)}}

\subsubsection{Estimates of roots of the recurrence polynomials}

\begin{proposition} \label{proposition: bound-abs-value-roots}
    For any word $w$ of the length $k\geq 9$, the corresponding recurrence polynomial $P_w(x)$ has one real "maximal" root in the interval $(q-q^{-k+2}, q)$ and $k-1$ roots counted up to multiplicity that have absolute values less than $1.65$.
\end{proposition}
\begin{proof}
We start by noticing that $P_w(x)$ has a real root close to $q$. Indeed, $P_w(q) = 1+ (q-q)A_w(q)=1 > 0$ and 
\begin{multline*}
    P_w(q-\frac{1}{q^{k-2}}) = 1 - \frac{1}{q^{k-2}} A_w(q-\frac{1}{q^{k-2}})\leq 1 - \frac{1}{q^{k-2}} \left(q-\frac{1}{q^{k-2}}\right)^{k-1} = \\
    = 1 - \frac{1}{q^{k-2}} q^{k-1}\left(1-\frac{1}{q^{k-1}}\right)^{k-1} 
    < 1 - q\left(1-(k-1)\frac{1}{q^{k-1}}\right) \leq 0
\end{multline*},
since $\frac{k-1}{q^{k-1}}\leq\frac{1}{2}$ for $k\geq3$, and $q\geq 2$. Hence, there is a real root of $P_w(x)$ in the interval $(q-\frac{1}{q^{k-2}},q)$. In Section \ref{subsec:Estimate of the root} we will obtain a better estimate of this root by taking into account  the autocorrelation of $w$.

Consider now the other roots of $P_w(x)$. We will use a standard technique of applying the Rouché theorem (see below) to different parts of the polynomial.
    \begin{theorem*}[Rouché]
        Let $\gamma$ be a closed contour and $D$ be the region enclosed by $\gamma$. Let also $f$ and $g$ be complex-valued functions which are holomorphic in $D$. Let $|g(z)| < |f(z)|$ on $\gamma$. Then $f$ and $f + g$ have the same number of zeroes in $D$ counted up to their multiplicity.
    \end{theorem*}

Consider $(x-b_{k-1})A_w(x)$\footnote{This particular trick can be found in the discussion of a somewhat similar subject in \cite{math.stackexchange}}, then 
\begin{equation*}
(x-b_{k-1})A_w(x)= x^k + \left[\sum_{i=1}^{k-2} (b_{i} - b_{k-1}b_{i+1})x^i\right] - b_{k-1}b_1.
\end{equation*}

One can notice that $- b_{k-1}b_1$ as well as $(b_{i} - b_{k-1}b_{i+1})$ for $1\leq i\leq k-2$ can only assume the values $-1$, $0$ or $1$. Fix  now $r=1.65$. Then for all $x$, such that $|x|=r$, the following relation holds.
\begin{equation*}
    |(x-b_{k-1})A_w(x)| \geq r^k - \sum_{i=0}^{k-2} r^i = r^k - \frac{r^{k-1} - 1}{r-1} = \frac{r^{k-1}(r^2-r-1)+1}{r-1}.
\end{equation*}
Hence,

\begin{equation*}
    |(x-q)A_w(x)| \geq \frac{r^{k-1}(r^2-r-1)+1}{r-1}\frac{|x-q|}{|x-b_{k-1}|}\geq \frac{r^{k-1}(r^2-r-1)+1}{r-1}\frac{q-r}{r+1}.
\end{equation*}

One can check that, if $k=9$ and $r=1.65$, then the following relation  holds

\begin{equation*}
    \frac{r^{k-1}(r^2-r-1)+1}{r-1}\frac{q-r}{r+1}\geq\frac{r^{k-1}(r^2-r-1)+1}{r-1}\frac{2-r}{r+1} = 1.01251 > 1.
\end{equation*}

For the chosen $r=1.65$, the value of $r^2-r-1 = 0.0725$ is positive. Hence, the expression $\frac{r^{k-1}(r^2-r-1)+1}{r-1}\frac{q-r}{r+1}$ is a monotonically increasing function of $k$. That means that the above inequality holds for all $k\geq9$, and therefore it holds for $k\geq9$

\begin{equation*}
    |(x-q)A_w(x)| > 1 \text{ for x with $|x|=r$.}
\end{equation*}

The Rouché theorem implies that $P_w(x) = 1+ (x-q)A_w(x)$ has the same number of roots within the open disk $|x|<r$ as $(x-q)A_w(x)$.

Observe now that the coefficients of $A_w(x)$ have values in $\{0,1\}$. It is known \cite{Odlyzko1993} that all such polynomials have roots bounded in absolute value by the golden ratio $\frac{1+\sqrt{5}}{2}<1.65$. Therefore, $(x-q)A_w(x)$ has only one root $x=q$ with absolute value greater than or equal to $r=1.65$. By the Rouché theorem, the same is true for $P_w(x)$.

\end{proof}

\begin{remark}
    Numerical simulations for all polynomials $P_w(x)$ with degrees $k\leq 20$ and for different values of $q$ demonstrate that, in fact, absolute values of all roots, except for the largest one, are bounded by the golden ratio. However, proving this would require more involved analysis, while for our purposes a bound of $1.65$ is sufficient.
\end{remark}

\subsubsection{Estimates of coefficients in front of exponential terms}

\begin{proposition} \label{proposition: bound on coefs other exponents}
    Consider a root $\alpha_i$ with multiplicity $m_i$ of $P_w(x)$, which is not the maximal root. Then for $1\leq j \leq m_i$ the absolute value of a non-zero coefficient $c_{ij}$ has the following bound.
    \begin{equation*}
        |c_{ij}|\leq 2^{k^2+k}\cdot e^{3k^2/2 + 2k}\cdot q^{k^3+k}\cdot k^{3k^3/2+3k^2+2k}.
    \end{equation*}
\end{proposition}

\begin{proof}
    By using the residue method (described, for instance, in \cite{henrici1988applied}), one gets the following expression for the coefficient $c_{ij}$ in front of $\frac{1}{(x-\alpha_i)^j}$ in the Partial fraction decomposition of $\frac{1}{P_w(x)}$:

    \begin{equation} \label{eq:formula-for-cij}
        c_{ij} = \frac 1 {(m_i-j)!}\lim_{x\to \alpha_i}\frac{d^{m_i-j}}{dx^{m_i-j}} \left[\frac{(x-\alpha_i)^{m_i}}{P_w(x)}\right]
    \end{equation}

    Denote $Q(x):=P_w(x)/(x-\alpha_i)^{m_i}$. We need to estimate the value of $\left(\frac{1}{Q}\right)^{(m_i-j)}(\alpha_i)$, i.e. the ($m_i-j$)-th derivative of $\frac{1}{Q(x)}$ at $\alpha_i$. As the first step, we  will be using the following  Faà di Bruno's formula.

    \begin{lemma*}[Faà di Bruno's formula]
        Given two differentiable functions $f$ and $g$ defined in a neighborhood of $x\in\mathbb{R}$ ($x\in\mathbb{C}$), $t\in\mathbb{N}$ and supposing that all necessary derivatives of $f(\cdot)$ at $g(x)$ and of$g(\cdot)$ at $x$ exist, there is the following expression for the $t$-th derivative of the function $f(g(x))$.
        \begin{equation*}
        \frac{d^t}{dx^t} f(g(x))=\sum \frac{t!}{r_1!\,1!^{r_1}\,r_2!\,2!^{r_2}\,\cdots\,r_t!\,t!^{r_t}}\cdot 
        f^{(r_1+\cdots+r_t)}(g(x))\cdot \prod_{l=1}^t\left(g^{(l)}(x)\right)^{r_l},
        \end{equation*}
where the sum is taken over all $t$-tuples of nonnegative integers $(r_1,\ldots,r_t)$ satisfying the following constraint

\begin{equation*}
    1\cdot r_1+2\cdot r_2+3\cdot r_3+\cdots+t\cdot r_t=t.
\end{equation*}
    \end{lemma*}

Let $t=m_i-j$. Then, by applying Faà di Bruno's formula to $f(x)=\frac{1}{x}$ and $g(x)=Q(x)$, we obtain

\begin{equation} \label{eq:cij_est 1}
    \lim_{x\to \alpha_i}\frac{d^{m_i-j}}{dx^{m_i-j}} \left[\frac{1}{Q(x)}\right] = \sum_{ 1\cdot r_1+2\cdot r_2+3\cdot r_3+\cdots+t\cdot r_t=t} \frac{t!(-1)^{r_1+\cdots+r_t}(r_1+\cdots+r_t)!}{r_1!\,1!^{r_1}\,r_2!\,2!^{r_2}\,\cdots\,r_t!\,t!^{r_t} Q(\alpha_i)^{r_1+\cdots+r_t+1}}\cdot \prod_{l=1}^t\left(Q^{(l)}(\alpha_i)\right)^{r_l},
\end{equation}

since $\left(\frac{1}{x}\right)^{(n)}=\frac{(-1)^n n!}{x^{n+1}}$.

Let us inspect the summation term for a given tuple of $r_1, r_2,\ldots, r_t$. We start getting an upper bound for $Q^{(l)}(\alpha_i)$.

\begin{lemma}
    For any $l\in\mathbb{N}$ the value of $Q^{(l)}(\alpha_i)$ is bounded by
    \begin{equation*}
        |Q^{(l)}(\alpha_i)| \leq qe^2\cdot k!
    \end{equation*}
\end{lemma}
\begin{proof}
    Observe that for any natural $r$
    \begin{equation*}
        \left|P_w^{(r)}(\alpha_i)\right|\leq \sum_{j=r}^{k} q\frac{j!}{(j-r)!}|\alpha_i|^{j-r}\leq q\cdot k!\sum_{j=r}^{k} \frac{|\alpha_i|^{j-r}}{(j-r)!}\leq q\cdot k!\sum_{j=0}^{\infty} \frac{2^{j}}{j!}\leq q\cdot k! e^2.
    \end{equation*}
    Here we used the fact that the coefficients of $P_w(x)$ are located between $-q$ and $q$ and (see Proposition \ref{proposition: bound-abs-value-roots}) the value of $|\alpha_i|$ is bounded by $1.65<2$.

    Notice now  that 

    \begin{equation*}
        \left|Q^{(l)}(\alpha_i)\right| \leq \left|\left[(x-\alpha_i)^{m_i}Q(x)\right]^{(l+m_i)}\big|_{x=\alpha_i}\right| = \left|P_w^{(l+m_i)}(\alpha_i)\right|.
    \end{equation*}
\end{proof}

Therefore,

\begin{equation} \label{eq:cij_est 2}
    \left |\prod_{l=1}^t\left(Q^{(l)}(\alpha_i)\right)^{r_l}\right| \leq (qe^2k!)^{r_1+\ldots r_t}\leq (qe^2k!)^{t}
\end{equation}

 We will estimate now the following term

\begin{equation} \label{eq: factorials term}
    \frac{t!(-1)^{r_1+\cdots+r_t}(r_1+\cdots+r_t)!}{r_1!\,1!^{r_1}\,r_2!\,2!^{r_2}\,\cdots\,r_t!\,t!^{r_t} Q(\alpha_i)^{r_1+\cdots+r_t+1}}
\end{equation}
in the expression (\ref{eq:cij_est 1}). Since $r_1+\dots+r_t < t$, the next estimate on (\ref{eq: factorials term}) 
has the following form
\begin{equation*}
    \left|\frac{t!(-1)^{r_1+\cdots+r_t}(r_1+\cdots+r_t)!}{r_1!\,1!^{r_1}\,r_2!\,2!^{r_2}\,\cdots\,r_t!\,t!^{r_t} Q(\alpha_i)^{r_1+\cdots+r_t+1}}\right| \leq \frac{t!^2}{1\cdot\min\left(\left|Q(\alpha_i)\right|^{t},1\right)}
\end{equation*}

Hence, the absolute value of the summation term is bounded by

\begin{equation} \label{eq:cij_est 3}
    t!^2\left(\frac{qe^2k!}{\min\left(\left|Q(\alpha_i)\right|,1\right)}\right)^t
\end{equation}

In order to obtain a lower bound for $Q(\alpha_i)$, we will use the following result on the minimal distance between different roots of a polynomial (\cite{Minimal-poly-root-distance})

\begin{theorem*}[Rump] \label{th:sep(P)}
    Let $P(x)$ be an arbitrary integral polynomial (perhaps having multiple zeros) of degree $k$. Denote by $s=|P|_1$ the sum of the absolute values of coefficients of $P$,and by $sep(P)$ the minimal distance between two different roots of $P$. Then
\begin{equation*}
    sep(P) > \left[2 k^{k/2+2}(s+1)^k\right]^{-1}
\end{equation*}
\end{theorem*}

In our case, since the coefficients of $P_w(x)$ are bounded by $q$, we get the following lower bound

\begin{equation*}
    sep(P) > \left[2 k^{k/2+2}(q(k+1)+1)^k\right]^{-1} = \left[2 k^{k/2+2}(qk)^k(1+\frac{3}{qk})^k\right]^{-1} \geq \left[2 k^{k/2+2}(qk)^ke^{3/q}\right]^{-1}
\end{equation*}

Hence,

\begin{equation} \label{eq:cij_est 4}
    \left|Q(\alpha_i)\right|=\left|\prod_{\alpha_j\neq\alpha_i}(\alpha_i-\alpha_j)\right|\geq\left(sep(P)\right)^{k-m_i}\geq \min(1, sep(P))^k\geq\left[2 k^{k/2+2}(qk)^ke^{3/q}\right]^{-k}.
\end{equation}

The number of $t$-tuples of nonnegative integers $(r_1,\ldots,r_t)$ satisfying

\begin{equation*}
    1\cdot r_1+2\cdot r_2+3\cdot r_3+\cdots+t\cdot r_t=t
\end{equation*}    
is bounded by the number of $(t+1)$-tuples of nonnegative integers $(r_0, r_1,\ldots,r_t)$ satisfying
\begin{equation*}
    r_0+r_1 +\ldots+ r_t=t, 
\end{equation*}    
which is 
\begin{equation}\label{eq:cij_est 5}
    \binom{(t +1) + (t - 1)}{t}=\binom{2t}{t}\leq (2t)^t.
\end{equation}

Remembering that $t=m_i-j$, and by combining (\ref{eq:formula-for-cij}), (\ref{eq:cij_est 1}), (\ref{eq:cij_est 2}), (\ref{eq:cij_est 3}), (\ref{eq:cij_est 4}) and (\ref{eq:cij_est 5}), we obtain

\begin{equation*}
    |c_{ij}|\leq \frac{1}{t!}\cdot(2t)^t \cdot t!^2 \cdot (q e^2 k!)^t \cdot \left[2 k^{k/2+2}(qk)^ke^{3/q}\right]^{kt}.
\end{equation*}
Since $t=m_i-j< k$, 
\begin{equation*}
    |c_{ij}|\leq (qk)^k \cdot k! \cdot ( e^2 k!)^k \cdot \left[2 k^{k/2+2}(qk)^ke^{3/2}\right]^{k^2}\leq (qk)^k \cdot k^k \cdot (2 e^2 k^k)^k \cdot \left[2q^{k} k^{3k/2+2}e^{3/2}\right]^{k^2}.
\end{equation*}
Then, after simplifying the terms, we arrive at
\begin{equation*}
    |c_{ij}|\leq 2^{k^2+k}\cdot e^{3k^2/2 + 2k}\cdot q^{k^3+k}\cdot k^{3k^3/2+3k^2+2k}.
\end{equation*}
\end{proof}

\subsection{Leading exponent of \texorpdfstring{$h_w(n)$}{hw(n)}}

\subsubsection{Estimate of the root} \label{subsec:Estimate of the root}

The Proposition \ref{proposition: bound-abs-value-roots} implies that the polynomial $P_w(x)$ has one real root $\alpha_m$, which is close to $q$, and all other its  roots are bounded in absolute value by $1.65$. Now we provide a more accurate estimate for the value of $\alpha_m$:

\begin{proposition} \label{proposition: largest-root-value}
    Consider a word $w$ of length $k\geq2$ with an inverse autocorrelation polynomial $A_w(x)$ and the recurrence polynomial $P_w(x)$. Then the difference between $q$ and the maximal root $\alpha_m$ of $P_w(x)$ is equal to
    \begin{equation*}
        q-\alpha_m = \frac{1}{A_w(q)} + \frac{A'_w(q)}{A_w(q)^3} + O(k^2 q^{-3k+5}),
    \end{equation*}
    where the upper bound constant in $O(k^2 q^{-3k+5})$ does not depend on $k$, $q$, and on the autocorrelation of $w$.
\end{proposition}
\begin{proof}

Consider $\varepsilon = q-\alpha_m$. From the proof of the Proposition \ref{proposition: bound-abs-value-roots} it follows that $0\leq\varepsilon\leq q^{-(k-2)}$. Now,
\begin{equation*}
    P_w(\alpha_m)=P(q-\varepsilon) = 1 + (q-\varepsilon-2)A_w(q-\varepsilon) = 1 - \varepsilon A_w(q-\varepsilon)=0.
\end{equation*}

By using Taylor expansion around $x=q$, one gets

\begin{equation*} 
    1 - \varepsilon \left(A_w(q) - A_w'(q)\varepsilon + R(\varepsilon))\right)=0
\end{equation*}
i.e.
\begin{equation*} \label{eq:taylor-for-Aw}
    \varepsilon \left(A_w(q) - A_w'(q)\varepsilon)\right)= 1 - \varepsilon R(\varepsilon),
\end{equation*}
where 

\begin{equation*}
    |\varepsilon R(\varepsilon)| \leq \varepsilon\frac{A_w'' (q)}{2!} \varepsilon^2 \leq  \varepsilon^3\sum_{i=0}^{k-3} k^2 q^i \leq k^2 q^{k-2} \varepsilon^3\leq k^2 q^{k-2} q^{-3(k-2)}=k^2 q^{-2k+4}.
\end{equation*}

The first inequality above follows from the Lagrange form of the remainder in Taylor expansion and from the fact that $A_w''(x)$ is an increasing function for positive $x$. 

Next, for small $0\leq\varepsilon\leq q^{-(k-2)}$ the function $\varepsilon \left(A_w(q) - A_w'(q) \varepsilon)\right)$ is monotonously increasing since
\begin{equation*}
    \frac{d [\varepsilon \left(A_w(q) - A_w'(q)\varepsilon)\right)]}{d \varepsilon} = A_w(q)-2A'_w(q)\varepsilon > q^{k-1} -2kq^{k-1}\cdot q^{-k+2}= q^{k-1} - 2kq > 0
\end{equation*}
for $k\geq6$ and $q\geq2$.
It means that $\varepsilon$ lies between the solutions of 

\begin{equation*}
    \varepsilon \left(A_w(q) - A_w'(q)\varepsilon)\right)= 1 - k^2 q^{-2k+4}
\end{equation*}
    and the solution of 
\begin{equation*}
    \varepsilon \left(A_w(2) - A_w'(2)\varepsilon)\right)= 1 + k^2 q^{-2k+4},
\end{equation*}
i.e.
\begin{equation*}
    \frac{A_w(q)-\sqrt{A_w(q)^2-4A'_w(q)(1-k^2 q^{-2k+4})}}{2A'_w(q)} \leq \varepsilon \leq \frac{A_w(q)-\sqrt{A_w(q)^2-4A'_w(q)(1+k^2 q^{-2k+4})}}{2A'_w(q)}.
\end{equation*}

Here, we take the solution with the minus sign in front of  the square brackets, since the solution with the plus sign is greater than

\begin{equation*}
    \frac{A_w(q)}{2 \cdot k A_w(q)/q}=\frac{q}{k},
\end{equation*}
and it can not satisfy the conditions $\varepsilon<q^{-(k-2)}$ for $k\geq 2$ and $q\geq 2$.

We will use the following bounds on the values of $\sqrt{1-x}$ for $0<x\leq 1$

\begin{equation*}
    1-\frac{x}{2}-\frac{x^2}{8}-\frac{x^3}{2}<\sqrt{1-x}<1-\frac{x}{2}-\frac{x^2}{8}.
\end{equation*}
Therefore,

\begin{multline*}
    \varepsilon \geq  \frac{A_w(q)-A_w(q)\sqrt{1-4A'_w(q)(1-k^2 q^{-2k+4})/A_w(q)^2}}{2A'_w(q)}\geq \\
    \geq\frac{A_w(q)-A_w(q) +2A'_w(q)(1-k^2 q^{-2k+4})/A_w(q) + \left(4A'_w(q)(1-k^2 q^{-2k+4})\right)^2/[8A_w(q)^3]}{2A'_w(q)}=\\
    =\frac{1-k^2 q^{-2k+4}}{A_w(q)}+\frac{A'_w(q)(1-k^2 q^{-2k+4})^2}{A_w(q)^3} \geq\\
    \geq \frac{1-k^2 q^{-2k+4}}{A_w(q)}+\frac{A'_w(q)(1-2k^2 q^{-2k+4})}{A_w(q)^3}
    \geq \frac{1}{A_w(q)} + \frac{A'_w(q)}{A_w(q)^3} - \left[\frac{k^2 q^{-2k+4}}{A_w(q)}+\frac{A'_w(q)2k^2 q^{-2k+4}}{A_w(q)^3}\right].
\end{multline*}
The above Bernoulli's inequality holds since $k^2 q^{-2k+4}\leq 1$ for $k\geq 4$ and $q\geq2$. Observe now that $A_w(q)\geq q^{k-1}$ and $A_w'(q)\leq k A_w(q)/q$. Thus, the last term in the square brackets is bounded by
\begin{equation*}
    k^2 q^{-2k+4}\cdot q^{-k+1} + k \cdot 2 k^2 q^{-2k+4}\cdot q^{-2k+2 - 1} = k^2 q^{-3k+5} + 2 k^3 q^{-4k+5} = O(k^2 q^{-3k+5}).
\end{equation*}
Here for $k \geq 2$ the exponents with $q$ divided by $q^{-3k+5}$ have non-positive powers. Therefore, if 

\begin{equation*}
  \frac{k^2 q^{-3k+5} + 2 k^3 q^{-4k+6}}{k^2 q^{-3k+5}} < C
\end{equation*}

for some $C>0$, $k\geq 2$ and $q=2$, then the same is true for all $q\geq 2$ and $k\geq2$.  Hence, a bound in $O(k^2 q^{-3k+5})$ on the coefficient does not depend on $q$ and on the autocorrelation of $w$.

Now, by using the other bound for $\sqrt{1-x}$ one gets similarly 

\begin{multline*}
    \varepsilon \leq  \frac{1+k^2 q^{-2k+4}}{A_w(q)}+\frac{A'_w(q)(1+k^2 q^{-2k+4})^2}{A_w(q)^3} + 
    \frac{1}{2}\frac{ \left(4A'_w(q)(1+k^2 q^{-2k+4})\right)^3/A_w(q)^5}{2A'_w(q)}\leq\\
    \leq\frac{1}{A_w(q)} + \frac{A'_w(q)}{A_w(q)^3} + \left[\frac{k^2 q^{-2k+4}}{A_w(q)}+\frac{A'_w(q)\left((1+k^2 q^{-2k+4})^2-1\right)}{A_w(q)^3}+\frac{16A'_w(q)^2(1+k^2 q^{-2k+4})^3}{A_w(q)^5}\right].
\end{multline*}

Again, since $A_w(q)\geq q^{k-1}$,  and $A_w'(q)\leq k A_w(q)/q$, the last term in the square brackets is bounded by

\begin{multline*}
    k^2 q^{-2k+4}\cdot q^{-(k-1)} + k\cdot q^{-2(k-1)-1}\cdot (1+k^2 q^{-2k+4}-1)(1+k^2 q^{-2k+4}+1)  + 16 k^2\cdot q^{-3(k-1)-2} \cdot (1+k^2 q^{-2k+4})^3 =\\
    = k^2 q^{-3k+5} + k^3 q^{-4k+5}(2+k^2 q^{-2k+4}) + 16k^2 q^{-3k+1}(1+k^2 q^{-2k+4})^3 = O(k^2 q^{-3k+5}).
\end{multline*}
As in the previous case, for $k \geq 2$ in the expression above the exponents with $q$ divided by $q^{-3k+5}$ have non-positive power. Hence, if 

\begin{equation*}
  \frac{k^2 q^{-3k+5} + k^3 q^{-4k+5}(2+k^2 q^{-2k+4}) + 16k^2 q^{-3k+1}(1+k^2 q^{-2k+4})^3}{k^2 q^{-3k+5}} < C
\end{equation*}

for some $C>0$, $k\geq 2$ and $q=2$, then the same is true for all $q\geq 2$ and $k\geq2$, i.e. a bound on the coefficient in $O(k^2 q^{-3k+5})$ does not depend neither on $q$ nor the autocorrelation of $w$.
\end{proof}

\subsubsection{Estimate of the value of coefficients}

\begin{proposition} \label{proposition: coefficient-before-alpha_m}
    Consider a word $w$ of length $k$ with an inverse autocorrelation polynomial $A_w(x)$, and the recurrence polynomial $P_w(x)$. Denote by $\alpha_m$ the maximal root of $P_w(x)$, and let $c_w$ be the coefficient in front of$\frac{1}{(x-\alpha_i)}$ in partial fraction decomposition of $\frac{1}{P_w(x)}$. Then
    \begin{equation*}
        \frac{1}{c_w}= A_w(q) - 2(q-\alpha_m)A_w'(q) + O(k^2 q^{-k+2})
    \end{equation*}
    where the upper bound constant in $O(k^2 q^{-k+2})$ does not depend neither on $k$, $q$ nor on the autocorrelation of $w$.
\end{proposition}
\begin{proof}
It follows from Proposition \ref{proposition: bound-abs-value-roots} that $\alpha_m$ is a single root of $P_w(x)$. One can see from (\ref{eq:formula-for-cij}) that

\begin{equation*}
    c_w=\frac{1}{0!}\lim_{x\to \alpha_m}\left[\frac{(x-\alpha_i)}{P_w(x)}\right] = \frac{1}{P_w'(\alpha_m)}
\end{equation*}
By the definition of a derivative we get
\begin{equation*}
    P_w'(\alpha_m) := \lim_{x\to \alpha_m}\left[\frac{P_w(x) - P_w(\alpha_m)}{(x-\alpha_i)}\right] = \lim_{x\to \alpha_m}\left[\frac{P_w(x)}{(x-\alpha_i)}\right].
\end{equation*}
Let $\varepsilon=q-\alpha_m$. Then, in turn, by using Taylor's theorem with Lagrange form of the remainder, one obtains

\begin{multline*}
    P_w'(\alpha_m)  = \left(1+(x-q)A_w(x)\right)'\Big|_{x=\alpha_m}= A_w(\alpha_m) - (q-\alpha_m) A_w'(\alpha_m) = A_w(q-\varepsilon) - \varepsilon A_w'(q-\varepsilon) = \\
    = A_w(q) - \varepsilon A_w'(q) +\frac{A_w''(q-\nu_1)}{2!}\varepsilon^2 - \varepsilon A_w'(q) +\varepsilon A_w''(q-\nu_2)\varepsilon = A_w(q) - 2\varepsilon A_w'(q) + O( k^2 q^{-k+2})
\end{multline*}
 for some $\nu_1,\nu_2\in [0,\varepsilon]$. Here we used the fact that $A''_w(x)$ is an increasing function of $x$, and hence the remainder in the Taylor expansion for both $\nu_1$ and $\nu_2$ is bounded by 
\begin{equation*}
    A_w''(q-\nu)\varepsilon \leq {A''_w(q)}\varepsilon^2\leq {k^2 q^{k-2}} q^{-2(k-2)} = O(k^2 q^{-k+2}).
\end{equation*}

Observe that the coefficient in $O(k^2 q^{-k+2})$ does not depend on $q$ nor on the autocorrelation of $w$.
\end{proof}

Now we are ready to prove Theorem \ref{th:main}. We start by looking at the moment of time when the graphs of the main exponents of $h_w(n)$ and $h_{w'}$ intersect for two arbitrary words $w$ and $w'$ of the same length:

\subsection{Intersection of the graphs of the dominant exponents of \texorpdfstring{$h_w(n)$}{hw(n)} and \texorpdfstring{$h_{w'}(n)$}{hw'(n)}}
\begin{theorem} \label{th:main-exponents-intersection}
    For $k \geq k_0$ sufficiently large, consider two words $w$ and $w'$ of the same length $k$, where the autocorrelation of $w$ is larger than that of $w'$. Let $\alpha_m$ and $\alpha_{m}'$ denote the largest roots of the recurrence polynomials $P_w(x)$ and $P_{w'}(x)$, respectively, and let $c_w$ and $c_{w'}$ be the coefficients in the partial fraction decompositions of $P_w(x)$ and $P_{w'}(x)$ corresponding to the terms $\frac{1}{x-\alpha_m}$ and $\frac{1}{x-\alpha_{m}'}$, respectively. Then the exponential terms $c_w (\alpha_m)^{\,n}$ and $c_{w'}(\alpha_{m}')^{\,n}$ intersect exactly once at

    \begin{equation*}
        n^* = C_{A_w/A_{w'}} \cdot q A_w(q)\cdot \left(1+O(k q^{-k/3+4})\right) \geq \ln q \cdot q^{k-1}\cdot \left(1+O(k q^{-k/3+4})\right)
    \end{equation*}
    where 
    \begin{equation*}
    C_{A_w/A_{w'}} = \ln\left(1 + \frac{A_{w}(q)-A_{w'}(q)}{A_{w'}(q)}\right)\Big/\left(\frac{A_{w}(q)-A_{w'}(q)}{A_{w'}(q)}\right)
    ~~\text{ and }~~ 
    \frac{\ln q}{q} \leq C_{A_w/A_{w'}}\leq 1.
    \end{equation*}
where $k_0$ does not depend on $q$, and the constant in the bound $O(k q^{-k/3+4})$ does not depend on $k$, $q$, and on the autocorrelations of $w$ and $w'$.
    
\end{theorem}
\begin{proof}
We begin by noting that graphs of any two different positive exponential functions intersect exactly once. In order to find the corresponding value of $n$, we start with the following (simple) equation:
\begin{equation*}
    c_w (\alpha_m)^n = c_{w'} (\alpha_m')^n.
\end{equation*}
Hence,
\begin{equation*}
    \frac{c_{w'}}{c_w} =  \left(\frac{\alpha_m}{\alpha_m'}\right)^n,
\end{equation*}
and
\begin{equation*}
    n = \frac{\ln(c_{w'}/c_{w})}{\ln(\alpha_m/\alpha_m')}. 
\end{equation*}
Thus, we need to analyze the ratios $c_{w'}/c_{w}$ and $\alpha_m/\alpha_{m}'$. We consider at first $\alpha_m/\alpha_{m}'$, and then estimate $c_{w'}/c_{w}$. Besides, we always assume that $k \geq 9$.
\newline

\subsubsection{Finding \texorpdfstring{$\mathbf{\alpha_m/\alpha_m'}$}{alpham/alpham'}:} 
\begin{proposition} \label{prop: alpha/alpha'}
   Under the same conditions as in Theorem~\ref{th:main-exponents-intersection}, the ratio
\begin{equation}
    \frac{\alpha_m}{\alpha_m'} 
    = 1 + \frac{A_{w}(q) - A_{w'}(q)}{q\,A_{w'}(q)\,A_{w}(q)}\left(1 + O(k q^{-k/3+4})\right),
\end{equation}
where the  (implicit) constant in $O(k q^{-k/3+4})$ does not depend on $q$, $k$, and on the autocorrelations of $w$ and $w'$.
\end{proposition}
\begin{proof}
Let $\varepsilon=q-\alpha_m$ and $\varepsilon'=q-\alpha_m'$. Then
\begin{equation} \label{eq:alpha-ratio-epsilon}
    \frac{\alpha_m}{\alpha_m'} = \frac{q-\varepsilon}{q-\varepsilon'}= \frac{1-\varepsilon/q}{1-\varepsilon'/q} = (1-\varepsilon/q)[1+\varepsilon'/q + (\varepsilon')^2/q^2+ O((\varepsilon'/q)^3)] = 1 + \frac{\varepsilon'-\varepsilon}{q}\left(1+\varepsilon'/q\right) + O\left(\frac{1}{q^{3(k-1)}}\right).
\end{equation}
Here we expanded $\frac{1}{1-\varepsilon'/q}$ as a geometric series and observed that
\begin{equation*}
    (\varepsilon'/q)^3 + (\varepsilon'/q)^4 +\ldots = \frac{(\varepsilon'/q)^3}{1-\varepsilon'/q}\leq \frac{(\varepsilon'/q)^3}{1-q^{-(k-1)}}\leq\frac{(\varepsilon'/q)^3}{1-2^{-(2-1)}} = 2(\varepsilon'/q)^3 \leq 2\left(q^{-(k-1)}\right)^3
\end{equation*}
since $\varepsilon < q^{-(k-2)}$.

Our next step is to understand how $\varepsilon' - \varepsilon$ is related to the values of $A_w(q)$, $A_w'(q)$, $A_{w'}(q)$, and $A_{w'}'(q)$. Making use of the Proposition~\ref{proposition: largest-root-value}, we obtain
\begin{multline} \label{eq: delta epsilon}
    \varepsilon'-\varepsilon=\frac{1}{A_{w'}(q)}-\frac{1}{A_w(q)} + \frac{A'_{w'}(q)}{A_{w'}(q)^3} - \frac{A'_w(q)}{A_w(q)^3} + O(k^2 q^{-3k+5}) =\\
    =\frac{A_{w}(q)-A_{w'}(q)}{A_{w'}(q)A_{w}(q)} +  \frac{A_{w}(q)^3 A_{w'}'(q)-A_{w'}(q)^3 A_{w}'(q)}{A_{w'}(q)^3A_{w}(q)^3} + O(k^2 q^{-3k+5}).
\end{multline}

We will show now that, due to the structure of the autocorrelations of $w$ and $w'$, the first term always dominates the second one:

\begin{proposition} \label{proposition: delta A_w dominates delta A'_w}
Let $A_w(x)$ and $A_{w'}(x)$ be the autocorrelation polynomials of two words $w$ and $w'$ of length $k$, respectively. Then
\begin{equation} \label{eq:delta A'= o(delta A)}
    \frac{A_{w}(q)^3 A_{w'}'(q)-A_{w'}(q)^3 A_{w}'(q)}{A_{w'}(q)^3A_{w}(q)^3} = O\left(kq^{-k/3 + 4}\frac{A_{w}(q)-A_{w'}(q)}{A_{w'}(q)A_{w}(q)} \right).
\end{equation}
    
\end{proposition}
\begin{proof}
    Let $\{b_i\}_{i=1}^{k}$ and $\{b'_i\}_{i=1}^{k}$ be the autocorrelations of $w$ and $w'$, respectively, and $i_{\neq}$ be the largest $i$, such that $b_i \neq b'_i$. We will consider two cases:
    \begin{enumerate}[(i)]
        \item $i_{\neq} \leq 2k/3 + 1$;
        \item $i_{\neq} > 2k/3 + 1$,
    \end{enumerate}
    which will be analyzed separately.
    \newline
    
    (i) In the first case, $A_{w}(q)-A_{w'}(q) = O(q^{2k/3 + 2})$, and $ A_{w}'(q)-A_{w'}'(q) = O(k q^{2k/3+1})$. Hence,

    \begin{multline*}
        \frac{A_{w}(q)^3 A_{w'}'(q)-A_{w'}(q)^3 A_{w}'(q)}{A_{w'}(q)^3A_{w}(q)^3} = 
        \frac{A_{w}(q)^3 (A_{w}'(q) + O(k q^{2k/3+1}))-[A_{w}(q)+O(q^{2k/3+2})]^3 A_{w}'(q)}{A_{w'}(q)^3A_{w}(q)^3} = \\
        =\frac{A_{w}(q)^3 A_{w}'(q) + A_{w}(q)^3 O(k q^{2k/3+1})-\left[A_{w}(q)^3+3A_{w}(q)^2 O(q^{2k/3+2}) + 3A_{w}(q) O(q^{4k/3+4}) + O(q^{6k/3 + 6})\right] A_{w}'(q)}{A_{w'}(q)^3 A_{w}(q)^3}.
    \end{multline*}
    Observe that $q^{k-1} \leq A_w(q) < q^k$ and $A_w'(q) \leq k A_w(q)/q$, and the same bounds hold for $w'$.  Assume that $k \geq 9$. Then the expression above is equal to
        \begin{multline*}\frac{O(k q^{2k/3+1})}{A_{w'}(q)^3} + \frac{3 O(q^{2k/3+2})A_{w}'(q)}{A_{w'}(q)^3 A_{w}(q)} + \frac{3O(q^{4k/3+4})A_{w}'(q)}{A_{w'}(q)^3 A_{w}(q)^2} + \frac{O(q^{6k/3 + 6})A_{w}'(q)}{A_{w'}(q)^3 A_{w}(q)^3}=\\
        ={O(k q^{2k/3+1}q^{-3(k-1)})} + {O(q^{2k/3+2}kq^{-3(k-1)})} + {O(q^{4k/3+4}kq^{-4(k-1)})}+ {O(q^{6k/3 + 6}kq^{-5(k-1)})}=\\
        = {O(k q^{-7k/3+4})} + {O(kq^{-7k/3+5})} + {O(kq^{-8k/3+8})}+ {O(kq^{-9k/3 + 11})}= O(k q^{-2k-k/3 + 4}).
    \end{multline*}
    Further, since $w$ and $w'$ have different autocorrelations, we get $A_{w}(q) - A_{w'}(q) \geq 1$, and
    \begin{equation*}
        \frac{A_{w}(q)-A_{w'}(q)}{A_{w'}(q)A_{w}(q)} > \frac{1}{q^{2k}}.
    \end{equation*}
    This gives us the bound~(\ref{eq:delta A'= o(delta A)}).
\newline

    (ii) For the second case, we will use the fact that $b_i$ and $b_i'$ are taken from the autocorrelations of $w$ and $w'$. Therefore, the values of $b_i$ cannot be arbitrary (see, for instance, Proposition~4.1 and Corollary~4.1 in \cite{BoldingBunimovich2019}):

    \begin{lemma} \label{lemma: conditions-on-b_i}
        Given a word $w$, let $b_k b_{k-1} \ldots b_1$ be its autocorrelation. Define $i_{m} = \max\{i : b_i = 1, \; 1 \leq i \leq k-1\}$. Then, for any $1 \leq i \leq k-1$ such that $b_i = 1$, either $i < k - i_{m}$ or $i = k - m \cdot (k - i_{m})$ for some integer $m$. In addition, if for some $1 \leq i \leq k-1$ there exists an integer $m$, such that $i = k - m \cdot (k - i_{m})$, then $b_i = 1$.

    \end{lemma}

    Since $i_{\neq} > 2k/3 + 1$, the value of $i_{m}$ for $w$ is also larger than $2k/3 + 1$. 
    Suppose that $i_{m}'$ for $w'$ is less than $i_{m} - 1$. Then
    \begin{multline*}
        A_w(q)-A_{w'}(q) \geq (q^{k-1} + q^{i_{m}-1}) - (q^{k-1} + q^{i_{m}-3}+q^{i_{m}-4}+q^{i_{m}-5}+\ldots + q^1 + q^{0}) \geq\\
        \geq q^{i_{m}-1} - q^{i_{m}-2} \geq q^{i_{m}-2} \geq q^{2k/3-1}.
    \end{multline*}
    Therefore,
    \begin{equation*} \label{eq:lower-bound-on-delta A / (Aw' Aw)}
        \frac{A_{w}(q)-A_{w'}(q)}{A_{w'}(q)A_{w}(q)} \geq \frac{q^{{2k/3-1}}}{q^k q^k} = q^{-2k+2k/3-1}.
    \end{equation*}
    In turn,
    \begin{equation}\label{eq:delta A ' = o(delta A) for large i_max}
        \frac{A_{w}(q)^3 A_{w'}'(q)-A_{w'}(q)^3 A_{w}'(q)}{A_{w'}(q)^3 A_{w}(q)^3} = \frac{ A_{w'}'(q)}{A_{w'}(q)^3 } -\frac{ A_{w}'(q)}{ A_{w}(q)^3}\leq 
        2\cdot k q^{-2(k-1)-1} =  O\left(k q^{-2k/3+2}\frac{A_{w}(q)-A_{w'}(q)}{A_{w'}(q)A_{w}(q)}\right)      
    \end{equation} 
    which provides a stronger estimate than (\ref{eq:delta A'= o(delta A)}).

    Finally, suppose that $i_{m}' = i_{m} - 1$. Then, since $k - i_{m} < k/3 - 1$, we have $k - 2(k - i_{m}) > k - 2k/3 + 2 > k/3 - 1 > k - i_{m}$. Then by Lemma~\ref{lemma: conditions-on-b_i}, the three largest values of $i$ such that $b_i = 1$ are $k$, $i_{m}$, and $k - 2(k - i_{m})$. 

    Similarly, for $w'$ we have $k - i_{m}' < k/3$ and $k - 2(k - i_{m}') > k - 2k/3 > k/3 > k - i_{m}'$. Then by Lemma~\ref{lemma: conditions-on-b_i}, the three largest values of $i$ such that $b_i' = 1$ are $k$, $i_{m} - 1$, and $k - 2(k - i_{m} + 1)$. Hence,
    \begin{multline*}
        A_w(q)-A_{w'}(q) \geq \\
        \geq (q^{k-1} + q^{i_{m}-1} + q^{k- 2(k-i_{m})-1}) - (q^{k-1} + q^{{i_{m}-2}}+ q^{k- 2(k-i_{m})-3}+q^{k- 2(k-i_{m})-4}+\ldots + q^1 + q^{0}) \geq\\
        \geq (q^{i_{m}-1} + q^{k- 2(k-i_{m})-1}) - (q^{{i_{m}-2}}+ q^{k- 2(k-i_{m})-2}) \geq q^{i_{m}-1} - q^{{i_{m}-2}}= q^{{i_{m}-2}}\geq q^{{2k/3-1}}.
    \end{multline*}
    Therefore,
     \begin{equation*}
        \frac{A_{w}(q)-A_{w'}(q)}{A_{w'}(q)A_{w}(q)} \geq \frac{q^{{2k/3-1}}}{q^k q^k} = q^{-2k+2k/3-1}
    \end{equation*}
    and the inequality~(\ref{eq:delta A ' = o(delta A) for large i_max}) holds again, which finishes a proof of the proposition.
\end{proof}
\paragraph{}

By returning to Eq.~(\ref{eq: delta epsilon}), we obtain

\begin{equation*} 
    \varepsilon'-\varepsilon=
    \frac{A_{w}(q)-A_{w'}(q)}{A_{w'}(q)A_{w}(q)}\left(1+O(k q^{-k/3+4})\right)+ O(k^2 q^{-3k+5}) = \frac{A_{w}(q)-A_{w'}(q)}{A_{w'}(q)A_{w}(q)}\left(1+O(k q^{-k/3+4})\right),
\end{equation*}
where the last equality holds due to the following relations
\begin{equation*}
    \frac{A_{w}(q)-A_{w'}(q)}{A_{w'}(q)A_{w}(q)} \geq q^{-2k},
\end{equation*}

and 
\begin{equation*}
    \frac{k^2 q^{-3k+5}}{q^{-2k}\cdot k q^{-k/3+4}}= \frac{k}{q^{2k/3-1}}<1,
\end{equation*}
where $k\geq9$ and $q\geq 2$.

By plugging the expression for $\varepsilon' - \varepsilon$ into (\ref{eq:alpha-ratio-epsilon}), we finally obtain

\begin{multline} \label{eq:alpha-ratio-final}
    \frac{\alpha_m}{\alpha_m'} = 1 + \frac{A_{w}(q)-A_{w'}(q)}{q A_{w'}(q)A_{w}(q)}\left(1+O(k q^{-k/3+4})\right)\left(1+O(q^{-k+1})\right) + O\left(q^{-3(k-1)}\right) = \\
    = 1 + \frac{A_{w}(q)-A_{w'}(q)}{q A_{w'}(q)A_{w}(q)}\left(1+O(k q^{-k/3+4})\right).
\end{multline}
We can perform the last transformation since $q^{-k+1} = O(k q^{-k/3+4})$, $k q^{-k/3+4} = O(1)$ for $k \geq 9$ and $q \geq 2$, and if 
$q^{-3(k-1)} = O\big(q^{-2k} \cdot k q^{-k/3+4}\big)$.

\end{proof}

\subsubsection{Finding \texorpdfstring{$\mathbf{{c_{w'}}/{c_w}}$}{cw'/cw}:} 

\begin{proposition} \label{prop: c_w'/c_w}
    Under the same conditions as in Theorem~\ref{th:main-exponents-intersection}, we have
    \begin{equation}
        \frac{c_{w'}}{c_{w}} = 1 + \frac{A_{w}(q) - A_{w'}(q)}{A_{w'}(q)}\left(1 + O\big(k q^{-k/3+4}\big)\right),
    \end{equation}
    where the implicit constant in $O(k q^{-k/3+4})$ does not depend on $q$, $k$, and on the autocorrelation of $w$ and $w'$.

\end{proposition}
\begin{proof}
Consider now the ratio of the coefficients $c_{w'}/c_w$.  By using Proposition~\ref{proposition: coefficient-before-alpha_m}, we obtain

\begin{multline*} 
    \frac{c_{w'}}{c_w}= \frac{A_w(q) - 2(q-\alpha_m)A_w'(q) + O(k^2 q^{-k+2})} {A_{w'}(q) - 2(q-\alpha_m ')A_{w'}'(q) + O(k^2 q^{-k+2})} = \frac{A_w(q)/A_{w'}(q) - 2\varepsilon A_w'(q)/A_{w'}(q) + O(k^2 q^{-2k+3})} {1 - 2\varepsilon' A_{w'}'(q)/A_{w'}(q) + O(k^2 2^{-2k+3})}.
\end{multline*}
We estimate now $\frac{1}{1 - 2 \varepsilon' A_{w'}'(q)/A_{w'}(q) + O(k^2 q^{-2k+3})}$ by expressing it as a geometric series.

\begin{equation*} 
\frac{1} {1 - 2\varepsilon' A_{w'}'(q)/A_{w'}(q) + O(k^2 2^{-2k+3})} = 
1 + 2\varepsilon' A_{w'}'(q)/A_{w'}(q) + O(k^2 q^{-2k+4}).
\end{equation*}
In more detail, since
\begin{equation*}
    2\varepsilon' A_{w'}'(q)/A_{w'}(q)\leq 2\cdot q^{-k+2}\cdot k = 2k q^{-k+2},
\end{equation*}
and $k^2 q^{-2k+4} = O(k q^{-k+2})$ for $q\geq2$ and $k\geq 9$, then

\begin{equation*}
     2\varepsilon' A_{w'}'(q)/A_{w'}(q) + O(k^2 q^{-2k+4}) \leq C \cdot k q^{-k+2}
\end{equation*}
for some constant $C$. Therefore,
\begin{multline*}
   \left|\sum_{n=2}^{\infty}\left(2\varepsilon' A_{w'}'(q)/A_{w'}(q)+ O(k^2 q^{-2k+3})\right)^n\right|\leq \sum_{n=2}^{\infty} (C k q^{-k+2})^n =\\
   =C^2k^2 q^{-2(k-2)}\cdot \frac{1}{1-C k q^{-k+2}} = C^2k^2 q^{-2(k-2)}\cdot \left(1+ O\left(C k q^{-k+2}\right)\right) = O(k^2q^{-2k+4})
\end{multline*}
for $k \geq k_0$, such that $C k_0 2^{-k_0+2} < 1$. Then the same bound holds for $q > 2$. Hence, we arrive at

\begin{equation} \label{eq:cw/cw'}
    \frac{c_{w'}}{c_w}= 
    \left(A_w(q)/A_{w'}(q) - 2\varepsilon A_w'(q)/A_{w'}(q) + O(k^2 q^{-2k+3})\right)\cdot
    \left(1 + 2\varepsilon' A_{w'}'(q)/A_{w'}(q) + O(k^2 q^{-2k+4})\right).
\end{equation}

Observe now that 

\begin{equation*}
    A_w(q)/A_{w'}(q)\leq q, ~~~ \varepsilon A_w'(q)/A_{w'}(q)\leq q^{-k+2}\cdot k ~~ \text{and} ~~ \varepsilon' A_{w'}'(q)/A_{w'}(q)\leq q^{-k+2}\cdot k,
\end{equation*}
and the last two expressions, as well as $k^2 q^{-2k+3}$, are less than $1$ for $k \geq 9$ and $q \geq 2$.  
Then, after multiplying and combining terms of the same order in (\ref{eq:cw/cw'}), we get

\begin{multline*}
    \frac{c_{w'}}{c_w}=  1 + \frac{A_{w}(q)-A_{w'}(q)}{A_{w'}(q)} - \frac{2\varepsilon A_w'(q)}{A_{w'}(q)} + O(k^2 q^{-2k+3})+ \frac{A_w(q)}{A_{w'}(q)}\frac{2\varepsilon' A_{w'}'(q)}{A_{w'}(q)} +O(k^2 q^{-2(k-2)}) +O(k q^{-2k+3}\cdot kq^{-k+2}) +\\
    +O(q\cdot k^2q^{-2k+4})= 1+ \frac{A_{w}(q)-A_{w'}(q)}{A_{w'}(q)} + 2\frac{\varepsilon' A_{w}(q)A_{w'}'(q) -\varepsilon A_{w'}(q)A_w'(q)}{A_{w'}(q)^2} + O(k^2 q^{-2k+5})
\end{multline*},
since $k^3 q^{-3k+5}<k^2 q^{-2k+5}$ for $q\geq2$ and $k\geq9$.

It follows from Proposition~\ref{proposition: largest-root-value}, that
\begin{equation*}
        \varepsilon = \frac{1}{A_w(q)} + \frac{A'_w(q)}{A_w(q)^3} + O(k^2 q^{-3k+5}) = 
        \frac{1}{A_w(q)} + O(kq^{-1}q^{-2(k-1)}) + O(k^2 q^{-3k+5}) = \frac{1}{A_w(q)}+ O(kq^{-2k+1)})
\end{equation*}
   for $k \geq 9$ and $q \geq 2$. The same estimate holds for $\varepsilon'$. Replacing $\varepsilon$ and $\varepsilon'$ with these estimates brings us to

\begin{multline} \label{eq: formula for cw'/cw with A' term}
    \frac{c_{w'}}{c_w}=  
    1 + \frac{A_{w}(q)-A_{w'}(q)}{A_{w'}(q)} + 2\frac{A_{w}(q)A_{w'}'(q)/A_{w'}(q) -A_{w'}(q)A_w'(q)/A_{w}(q)+O(q^k \cdot kq^{k-1}\cdot k q^{-2k+2})}{A_{w'}(2)^2} +\\
    +O(k^2 q^{-2k+5})
    =1 + \frac{A_{w}(q)-A_{w'}(q)}{A_{w'}(q)} + 2\frac{A_{w}(q)^2 A_{w'}'(q) -A_{w'}(q)^2 A_w'(q)}{A_{w}(q)A_{w'}(q)^3} + O(k^2 q^{-2k+5})
\end{multline},
since
\begin{equation*}
    \frac{q^k \cdot kq^{k-1}\cdot k q^{-2k+2}}{A_{w'}(2)^2} \leq \frac{k^2q}{q^{2(k-1)}}=k^2 q^{-2k+3}.
\end{equation*}
Similarly to Proposition~\ref{proposition: delta A_w dominates delta A'_w}, one can show that
\begin{proposition} \label{proposition: delta A_w dominates delta A'_w for cw/cw'}.
Let $A_w(x)$ and $A_{w'}(x)$ be the autocorrelation polynomials of two words $w$ and $w'$ of the same length $k$. Then
\begin{equation} \label{eq:delta A'= o(delta A) for coefficients}
    \frac{A_{w}(q)^2 A_{w'}'(q) -A_{w'}(q)^2 A_w'(q)}{A_{w}(q)A_{w'}(q)^3}= O\left(kq^{-k/3+4}\frac{A_{w}(q)-A_{w'}(q)}{A_{w'}(q)} \right).
\end{equation}
\end{proposition}
    The proof is almost the same as for the Proposition~\ref{proposition: delta A_w dominates delta A'_w}, with just a slight modification, which we present  here for completeness:
\begin{proof}
    Let $\{b_i\}_{i=1}^{k}$ and $\{b'_i\}_{i=1}^{k}$ be the autocorrelations of $w$ and $w'$, respectively, and $i_{\neq}$ be the largest $i$, such that $b_i \neq b'_i$. We will consider two cases:
    \begin{enumerate}[(i)]
    \item $i_{\neq} \leq 2k/3 + 1$; and
    \item $i_{\neq} > 2k/3 + 1$.
    \end{enumerate}
    ~
    \newline    
    (i) In the first case $A_{w}(q)-A_{w'}(q) = O(q^{2k/3 + 2})$, and $ A_{w}'(q)-A_{w'}'(q) = O(k q^{2k/3+1})$. Hence,

    \begin{multline*}
        \frac{A_{w}(q)^2 A_{w'}'(q)-A_{w'}(q)^2 A_{w}'(q)}{A_{w'}(q)^3A_{w}(q)} = 
        \frac{A_{w}(q)^2 (A_{w}'(q) + O(k q^{2k/3+1}))-[A_{w}(q)+O(q^{2k/3+2})]^2 A_{w}'(q)}{A_{w'}(q)^3A_{w}(q)} = \\
        =\frac{A_{w}(q)^2 A_{w}'(q) + A_{w}(q)^2 O(k q^{2k/3+1})-\left[A_{w}(q)^2+2A_{w}(q) O(q^{2k/3+2})+ O(q^{4k/3 + 4})\right] A_{w}'(q)}{A_{w'}(q)^3 A_{w}(q)}.
    \end{multline*}
    Notice that $q^{k-1} \leq A_w(q) < q^k$ and $A_w'(q) \leq k A_w(q)/q$. The same bounds hold for $w'$. We will assume $k \geq 9$. Then the expression above reads as
        \begin{multline*}\frac{A_w(q)O(k q^{2k/3+1})}{A_{w'}(q)^3} + \frac{2 O(q^{2k/3+2})A_{w}'(q)}{A_{w'}(q)^3} + \frac{O(q^{4k/3 + 4})A_{w}'(q)}{A_{w'}(q)^3 A_{w}(q)}=\\
        ={O(q^{k}\cdot k \cdot q^{2k/3+1}q^{-3(k-1)})} + {O(q^{2k/3+2}\cdot kq^{-2(k-1)})} + {O(q^{4k/3 + 4}\cdot kq^{-3(k-1)})}=\\
        = {O(k q^{-4k/3+4})} + {O(kq^{-4k/3+4})} + {O(kq^{-5k/3 + 7})}= O(k q^{-k-k/3 + 4})
    \end{multline*}
    for $k\geq9$ and $q\geq2$.

    At the same time, since $w$ and $w'$ have different autocorrelations, we have $A_{w}(q) - A_{w'}(q) \geq 1$, and
    \begin{equation*}
        \frac{A_{w}(q)-A_{w'}(q)}{A_{w'}(q)} > \frac{1}{q^{k}},
    \end{equation*}
    which gives the bound (\ref{eq:delta A'= o(delta A) for coefficients}).
\newline

    (ii) For the second case, we will use Lemma~\ref{lemma: conditions-on-b_i} again.

    Since $i_{\neq} > 2k/3 + 1$, then the value of $i_{m}$ (as in Lemma~\ref{lemma: conditions-on-b_i}) for $w$ is also greater than $2k/3 + 1$. 
    
    Suppose that $i_{m}'$ for $w'$ is less than $i_{m} - 1$. Then
    \begin{multline*}
        A_w(q)-A_{w'}(q) \geq (q^{k-1} + q^{i_{m}-1}) - (q^{k-1} + q^{i_{m}-3}+q^{i_{m}-4}+q^{i_{m}-5}+\ldots + q^1 + q^{0}) \geq\\
        \geq q^{i_{m}-1} - q^{i_{m}-2} \geq q^{i_{m}-2} \geq q^{2k/3-1}.
    \end{multline*}
    
    Therefore,
    \begin{equation*} \label{eq:lower-bound-on-delta A / (Aw' Aw) for cw/cw'}
        \frac{A_{w}(q)-A_{w'}(q)}{A_{w'}(q)} \geq \frac{q^{{2k/3-1}}}{q^k} = q^{-k+2k/3-1}.
    \end{equation*}
    In turn,
    \begin{multline}\label{eq:delta A ' = o(delta A) for large i_max for cw/cw'}
        \frac{A_{w}(q)^2 A_{w'}'(q)-A_{w'}(q)^2 A_{w}'(q)}{A_{w'}(q)^3 A_{w}(q)} = \frac{ A_{w}(q)A_{w'}'(q)}{A_{w'}(q)^3 } -\frac{ A_{w}'(q)}{ A_{w'}(q)A_{w}(q)}\leq \\
        \leq
        q^k\cdot k q^{-2(k-1)} + k q^{-1} \cdot q^{-k+1} \leq 2 kq^{-k+2}=  O\left(k q^{-2k/3+3}\frac{A_{w}(q)-A_{w'}(q)}{A_{w'}(q)A_{w}(q)}\right)    
    \end{multline},
    which gives a stronger estimate than (\ref{eq:delta A'= o(delta A)}).

    Assume finally that $i_{m}' = i_{m} - 1$. Then, since $k - i_{m} < k/3 - 1$, we have $k - 2(k - i_{m}) > k - 2k/3 + 2 > k/3 - 1 > k - i_{m}$. By Lemma~\ref{lemma: conditions-on-b_i}, the three largest values of $i$, such that $b_i = 1$, are $k$, $i_{m}$, and $k - 2(k - i_{m})$. 

    Similarly, for $w'$ we have $k - i_{m}' < k/3$ and $k - 2(k - i_{m}') > k - 2k/3 > k/3 > k - i_{m}'$. Then by Lemma~\ref{lemma: conditions-on-b_i}, the three largest values of $i$, such that $b_i' = 1$, are $k$, $i_{m} - 1$, and $k - 2(k - i_{m} + 1)$. Hence,
    \begin{multline*}
        A_w(q)-A_{w'}(q) \geq \\
        \geq (q^{k-1} + q^{i_{m}-1} + q^{k- 2(k-i_{m})-1}) - (q^{k-1} + q^{{i_{m}-2}}+ q^{k- 2(k-i_{m})-3}+q^{k- 2(k-i_{m})-4}+\ldots + q^1 + q^{0}) \geq\\
        \geq (q^{i_{m}-1} + q^{k- 2(k-i_{m})-1}) - (q^{{i_{m}-2}}+ q^{k- 2(k-i_{m})-2}) \geq q^{i_{m}-1} - q^{{i_{m}-2}}= q^{{i_{m}-2}}\geq q^{{2k/3-1}}.
    \end{multline*}
    Therefore,
    \begin{equation*} 
        \frac{A_{w}(q)-A_{w'}(q)}{A_{w'}(q)} \geq \frac{q^{{2k/3-1}}}{q^k} = q^{-k+2k/3-1}.
    \end{equation*}
    and the inequality~(\ref{eq:delta A ' = o(delta A) for large i_max for cw/cw'}) holds again, which finishes a proof of the proposition.

\end{proof}
\paragraph{}

By returning to Eq.~(\ref{eq: formula for cw'/cw with A' term}), and by using this bound, we obtain our final estimate for the ratio of the coefficients
\begin{equation} \label{eq: c_w ratio final}
    \frac{c_{w'}}{c_w}=  
    1 + \frac{A_{w}(q)-A_{w'}(q)}{A_{w'}(q)}\left(1+O\left(kq^{-k/3+4} \right)\right).
\end{equation}
The term $O(k^2 q^{-2k+5})$ is absorbed into $O\left(kq^{-k/3+4}\frac{A_{w}(q)-A_{w'}(q)} {A_{w'}(q)} \right)$ since
\begin{equation*}
    \frac{A_{w}(q)-A_{w'}(q)} {A_{w'}(q)}\geq \frac{1}{q^{k-1}}
\end{equation*},
and

\begin{equation*}
    \frac{k^2q^{-2k+5}}{kq^{-k/3+4}q^{-k+1}}=k q^{-2k/3}<1
\end{equation*}
for $k\geq9$ and $q\geq2$.

\end{proof}

\subsubsection{Finding the intersection time \texorpdfstring{$\mathbf{n^*}$}{n*}:} 

Finally, we turn our attention to finding $n^*$. Recall that $n^* = \frac{\ln(c_{w'}/c_{w})}{\ln(\alpha_m / \alpha_m')}$. Then by using (\ref{eq:alpha-ratio-final}) and (\ref{eq: c_w ratio final}), we obtain the following relation

\begin{equation*}
    n^{*} =\ln\left(
    1 + \frac{A_{w}(q)-A_{w'}(q)}{A_{w'}(q)}\left(1+O\left(kq^{-k/3 +4} \right)\right)\right)\Big/\ln\left(1 + \frac{A_{w}(q)-A_{w'}(q)}{qA_{w'}(q)A_{w}(q)}\left(1+O(k q^{-k/3+4})\right)\right).
\end{equation*}
Consider, at first, the numerator $\ln(c_{w'}/c_{w})$.
\begin{multline} \label{eq:ln(c_w'/c_w)_long}
   \ln(c_{w'}/c_{w})=\ln\left( 1 + \frac{A_{w}(q)-A_{w'}(q)}{A_{w'}(q)}\left(1+O\left(kq^{-k/3+4} \right)\right)\right) = \\
   =\left[\ln\left(1 + \frac{A_{w}(q)-A_{w'}(q)}{A_{w'}(q)}\right)\Big/\left(\frac{A_{w}(q)-A_{w'}(q)}{A_{w'}(q)}\right)\right] \cdot \left(\frac{A_{w}(q)-A_{w'}(q)}{A_{w'}(q)}\right)\cdot\\
   \cdot\left[\ln\left(1 + \frac{A_{w}(q)-A_{w'}(q)}{A_{w'}(q)}\left(1+O\left(kq^{-k/3+4} \right)\right)\right)\Big/\ln\left(1 + \frac{A_{w}(q)-A_{w'}(q)}{A_{w'}(q)}\right)\right] 
\end{multline}

We examine now more closely the ratio of the two logarithms, i.e., the last term in (\ref{eq:ln(c_w'/c_w)_long}). Observe that the value of the term $\left(1 + O\left(k q^{-k/3+4}\right)\right)$ lies between $1/2$ and $2$ for sufficiently large $k$, and that $\frac{A_{w}(q) - A_{w'}(q)}{A_{w'}(q)}$ is positive. Then, since for positive real $x$ and $y$ the following relation holds $|\ln(1 + x) - \ln(1 + y)| < |x - y|$, and we have

\begin{equation*}
    \left|\ln\left(1 + \frac{A_{w}(q)-A_{w'}(q)}{A_{w'}(q)}\left(1+O\left(kq^{-k/3+4} \right)\right)\right) - \ln\left(1 + \frac{A_{w}(q)-A_{w'}(q)}{A_{w'}(q)}\right)\right| \leq \frac{A_{w}(q)-A_{w'}(q)}{A_{w'}(q)}O\left(kq^{-k/3+4} \right)
\end{equation*}
Besides (\ref{eq:ln(c_w'/c_w)}) gets transformed to

\begin{equation*}\label{eq:ln(c_w'/c_w)}
    \ln(c_{w'}/c_{w}) = C_{A_w/A_{w'}} \cdot \left(\frac{A_{w}(q)-A_{w'}(q)}{A_{w'}(q)}\right)\cdot \left(1+O\left(kq^{-k/3+4} \right)\right)
\end{equation*},
where
\begin{equation*}
    C_{A_w/A_{w'}} = \ln\left(1 + \frac{A_{w}(q)-A_{w'}(q)}{A_{w'}(q)}\right)\Big/\left(\frac{A_{w}(q)-A_{w'}(q)}{A_{w'}(q)}\right).
\end{equation*}
It is easy to see that 
\begin{equation*}
    \frac{\ln q}{q} \leq C_{A_w/A_{w'}}\leq 1
\end{equation*},
since $0<({A_{w}(q)-A_{w'}(q)})/{A_{w'}(q)}\leq q$.

Consider now the denominator $\ln(\alpha_m/\alpha_m')$:

\begin{multline*}
    \ln(\alpha_m/\alpha_m') =  \ln\left(1 + \frac{A_{w}(q)-A_{w'}(q)}{qA_{w'}(q)A_{w}(q)}\left(1+O(k q^{-k/3+4})\right)\right) = \\
    = \frac{A_{w}(q)-A_{w'}(q)}{qA_{w'}(q)A_{w}(q)}\left(1+O(k q^{-k/3+4})\right) +
    O\left(\left[\frac{A_{w}(q)-A_{w'}(q)}{qA_{w'}(q)A_{w}(q)}\left(1+O(k q^{-k/3+4})\right)\right]^2\right) = \\
    = \frac{A_{w}(q)-A_{w'}(q)}{qA_{w'}(q)A_{w}(q)}\left(1+O(k q^{-k/3+4})\right) +
    O\left(\left[\frac{A_{w}(q)-A_{w'}(q)}{qA_{w'}(q)A_{w}(q)}\right]^2\right) = \frac{A_{w}(q)-A_{w'}(q)}{qA_{w'}(q)A_{w}(q)}\left(1+O(k q^{-k/3+4})\right)
\end{multline*}

Since 

\begin{equation*}
    \left[\frac{A_{w}(q)-A_{w'}(q)}{qA_{w'}(q)A_{w}(q)}\right]^2 \leq \left[\frac{A_{w}(q)-A_{w'}(q)}{qA_{w'}(q)A_{w}(q)}\right]\cdot \frac{1}{q^k} = o\left(\frac{A_{w}(q)-A_{w'}(q)}{qA_{w'}(q)A_{w}(q)}k q^{-k/3+4}\right),
\end{equation*}

We get finally

\begin{multline*}
    n^{*} = \left [C_{A_w/A_{w'}} \cdot \left(\frac{A_{w}(q)-A_{w'}(q)}{A_{w'}(q)}\right)\cdot \left(1+O\left(kq^{-k/3+4} \right)\right) \right]\Big / \left [ \frac{A_{w}(q)-A_{w'}(q)}{qA_{w'}(q)A_{w}(q)}\left(1+O(k q^{-k/3+4})\right)\right] =\\ =C_{A_w/A_{w'}} \cdot q A_w(q)\cdot \left(1+O(k q^{-k/3+4})\right) \geq C_{A_w/A_{w'}} \cdot q^{k}\cdot \left(1+O(k q^{-k/3+4})\right).
\end{multline*}
 It concludes a proof of Theorem \ref{th:main-exponents-intersection}.
\end{proof}

\subsection{Intersection of \texorpdfstring{$\mathbf{h_w(n)}$}{hw(n)} and \texorpdfstring{$\mathbf{h_{w'}(n)}$}{hw'(n)}}

The Theorem~\ref{th:main-exponents-intersection} shows that the graphs of the main exponents of $h_w(n)$ and $h_{w'}(n)$ intersect at a time $n$, which grows with $k$ exponentially. 

It remains to show that the other exponents, which contribute to $h_w(n)$ and $h_{w'}(n)$, do not alter the fact that the intersection time of these first-hitting probability curves still depends exponentially on $k$.

By returning to (\ref{eq: combinatoric expression of h_w(n)}), one can break the expression for $h_w(n)$ into two parts, where one corresponds to the main exponent, and the other one to all remaining exponents:

\begin{equation} \label{eq: expression for h_w(n)}
    h_w(n) = c_w (\alpha_m)^{n-1} + \sum_{\alpha_i\neq \alpha_m}\sum_{j=1}^{m_i} c_{ij}\binom{n-1}{j-1}(\alpha_i)^{n-j}
\end{equation}

It follows from Proposition~\ref{proposition: bound on coefs other exponents}, that for sufficiently large $k$

\begin{equation*}
\ln|c_{ij}| \leq (k^2+k) \ln 2 + \frac{3k^2}{2} + 2k + (k^3+k)\ln q + \left(\frac{3k^3}{2}+3k^2+2k\right) \ln k
< 2 k^3 (\ln q + \ln k)
\end{equation*}
Observe that the lower bound on such $k$ does not depend on $q$, since it is enough that the following relation holds

\begin{equation*}
 k^3 \ln(k)/2 > (k^2+k) \ln 2 + 3k^2/2 + 2k  + {(3k^2+2k)} \ln k.
\end{equation*}

By combining a bound on $|c_{ij}|$ with Proposition~\ref{proposition: bound-abs-value-roots}, we obtain the following bound on the second term in (\ref{eq: expression for h_w(n)}):

\begin{equation*}
   \sum_{\alpha_i\neq \alpha_m}\sum_{j=1}^{m_i} c_{ij}\binom{n-1}{j-1}(\alpha_i)^{n-j} 
   \leq \sum_{\alpha_i\neq \alpha_m}\sum_{j=1}^{m_i} k^{2k^3}q^{2k^3}\cdot n^k \cdot 1.65^{n-1} 
    \leq k\cdot k^{2k^3}q^{2k^3} n^k 1.65^{n-1}=k^{2k^3+1}q^{2k^3} n^k 1.65^{n-1}.
\end{equation*}

Next, we will show that

\begin{equation*}
    h_{w'}(N_{-}) - h_w(N_{-}) > c_{w'} (\alpha_m')^{N_{-}-1} - c_w (\alpha_m)^{N_{-}-1} - 2\cdot k^{2k^3+1}q^{2k^3} n^k 1.65^{n-1} > 0
\end{equation*}
 for an appropriate choice of $N_{-}$. Then it follows, by Theorem 2.1 in \cite{BoldingBunimovich2019}, since there is only one intersection of $h_w(n)$ with $h_{w'}(n)$, and the inequality $h_{w'}(n) < h_w(n)$ holds afterwards, that the intersection occurs for $n > N_{-}$.

Consider a given $0 < \delta < 0.25$, and let $N_{-} = (1 - 2\delta) \, C_{A_w/A_{w'}} \, q A_w(q)$, with $C_{A_w/A_{w'}}$ as in Theorem~\ref{th:main-exponents-intersection}, where $n^*$ is the intersection point of the graphs of the main exponents from that theorem. Then, for sufficiently large $k$, we have $n^* - N_{-} > \delta n^* + 1$, and  particularly,

\begin{equation*}
    1\cdot q\cdot q^{k+1}>N_{-} > (1-2\delta) \cdot \ln q \cdot q^{k-1} \geq 0.5 \cdot 0.5 \cdot 2^{k-1} = 2^{-2} q^{k-1}.
\end{equation*}

Therefore,
\begin{multline} \label{eq:delta cw alpham}
    c_{w'}(\alpha_m')^{N_{-}-1} - c_w (\alpha_m)^{N_{-}-1} = c_w (\alpha_m')^{N_{-}-1} \left(\frac{c_{w'}}{c_{w}} - \left(\frac{\alpha_m}{\alpha_m'}\right)^{N_{-}-1} \right) \geq
    c_w (\alpha_m')^{N_{-}-1} \left(\frac{c_{w'}}{c_{w}}  - \left(\frac{\alpha_m}{\alpha_m'}\right)^{(1-\delta)n^*}\right) = \\
    = c_w (\alpha_m')^{N_{-}-1} \left( \frac{c_{w'}}{c_{w}} - \left(\frac{c_{w'}}{c_{w}}\right)^{1-\delta} \right) = c_w (\alpha_m')^{N_{-}-1}\left(\frac{c_{w'}}{c_{w}}\right)^{1-\delta} \left(\left(\frac{c_{w'}}{c_{w}}\right)^{\delta} - 1  \right).
\end{multline}

In the proof of Theorem \ref{th:main-exponents-intersection} we obtained that

\begin{equation*}
    \frac{c_{w'}}{c_w}=  
    1 + \frac{A_{w}(q)-A_{w'}(q)}{A_{w'}(q)}\left(1+O\left(kq^{-k/3+4} \right)\right) \geq 1+\frac{1}{q^{k}}\left(1+O\left(kq^{-k/3+4} \right)\right)\geq 1+q^{-k-1}
\end{equation*}
for large enough $k$. Hence, for a fixed $\delta$, we have
\begin{equation*}
    \left(\frac{c_{w'}}{c_w}\right)^{\delta}\geq   
    1 + \delta q^{-(k+1)}+ O\left(q^{-2k+4}\right)\geq 1 + \delta q^{-(k+2)}
\end{equation*}
where $k$ is large enough. Also, for sufficiently large $k$ we have  $c_{w'}/c_w > 1$. Then $(c_{w'}/c_w)^{1-\delta} > 1$. Additionally, it follows from Proposition~\ref{proposition: coefficient-before-alpha_m}, that

\begin{equation*}
    {c_w}= \frac{1}{A_w(q) - 2(q-\alpha_m)A_w'(q) + O(k^2 q^{-k+2})} \geq \frac{1}{q^{k}}
\end{equation*}

if $k$. is large enough By substituting the obtained bounds into (\ref{eq:delta cw alpham}), we arrive at

\begin{equation*}
    c_{w'}(\alpha_m')^{N_{-}-1} - c_w (\alpha_m)^{N_{-}-1} \geq q^{-{k}}\cdot (\alpha_m')^{N_{-}-1}\cdot 1\cdot \delta q^{-(k+2)}.
\end{equation*}

Finally, combining the above estimates with the fact that 
\begin{equation*}
    \alpha_m' \geq q - q^{-k+2} \geq 0.99 \, q \geq 1.98 \geq 1.65 \cdot 1.2
\end{equation*}
for $k \geq 0$ and $q \geq 2$, we obtain, for sufficiently large $k$, the following relation

\begin{multline*}
    h_{w'}(N_{-}) - h_w(N_{-}) > q^{-k}\cdot (\alpha_m')^{N_{-}-1}\cdot \delta q^{-(k+2)} - 2k^{2k^3+1}q^{2k^3} (N_{-})^k  1.65 ^{N_{-}-1} =\\ 
    =1.65 ^{N_{-}-1}
    \left(\left(\frac{\alpha_m'}{1.65}\right)^{N_{-}-1}\cdot \delta q^{-(k+2)-k} 
    - \exp\left(\ln2+ ({2k^3+1})\ln k+({2k^3})\ln q +k \ln N_{-}\right)\right) \geq\\
    \geq 1.65 ^{N_{-}-1}\left(\exp\left(0.25 q^{k-1} \ln 1.2 + \ln \delta + ({-(k+2)-k})\ln q \right) \right. -\\
    \left. \exp\left(\ln2+ ({2k^3+1})\ln k+({2k^3})\ln q +k (k+2)\ln q\right)\right) > 0
\end{multline*}.

Analogously, for the upper bound on the intersection time,
\begin{equation*}
    N_{+} := (1 + 2\delta) \, C_{A_w/A_{w'}} \, 2 A_w(2) > N_{-} \geq 2^{-2} q^{k-1},
\end{equation*} 
we have $N_{+} > (1 + \delta) n^* + 1$, and

\begin{multline*}
   c_w (\alpha_m)^{N_{+}-1} - c_{w'}(\alpha_m')^{N_{+}-1}  = c_w (\alpha_m')^{N_{+}-1} \left(\left(\frac{\alpha_m}{\alpha_m'}\right)^{N_{+}-1} - \frac{c_{w'}}{c_{w}} \right) \geq\\
   \geq
    c_w (\alpha_m')^{N_{+}-1} \left(\left(\frac{\alpha_m}{\alpha_m'}\right)^{(1+\delta)n^*} - \frac{c_{w'}}{c_{w}}\right) = 
     c_w (\alpha_m')^{N_{+}-1} \left(\left(\frac{c_{w'}}{c_{w}}\right)^{1+\delta} - \frac{c_{w'}}{c_{w}}  \right) =\\
     =c_w (\alpha_m')^{N_{+}-1}\left(\frac{c_{w'}}{c_{w}}\right)\left(\left(\frac{c_{w'}}{c_{w}}\right)^{\delta} - 1  \right)\geq 
    c_w (\alpha_m')^{N_{+}-1}\left(\left(\frac{c_{w'}}{c_{w}}\right)^{\delta} - 1  \right)
\end{multline*}
for sufficiently large $k$. Then, exactly as in the previous case, one obtains

\begin{multline*}
    h_{w}(N_{+}) - h_{w'}(N_{+}) > q^{-k}\cdot (\alpha_m')^{N_{+}-1}\cdot \delta q^{-(k+2)} - 2k^{2k^3+1}q^{2k^3} (N_{+})^k  1.65 ^{N_{+}-1} =\\ 
    =1.65 ^{N_{+}-1}
    \left(\left(\frac{\alpha_m'}{1.65}\right)^{N_{+}-1}\cdot \delta q^{-(k+2)-k} 
    - \exp\left(\ln2+ ({2k^3+1})\ln k+({2k^3})\ln q +k \ln N_{+}\right)\right) \geq\\
    \geq 1.65 ^{N_{+}-1}\left(\exp\left(0.25 q^{k-1} \ln 1.2 + \ln \delta + ({-(k+2)-k})\ln q \right) \right. -\\
    \left. \exp\left(\ln2+ ({2k^3+1})\ln k+({2k^3})\ln q +k (k+2)\ln q\right)\right) > 0
\end{multline*}
for sufficiently large $k$. 
Finally, since the curves $h_w(n)$ and $h_{w'}(n)$ intersect just once, and $h_w(n) < h_{w'}(n)$ holds before that, the intersection must occur for $n < N_{+}$, which concludes a proof of Theorem~\ref{th:main}.

\bibliographystyle{alpha}
\bibliography{references}
\end{document}